\documentclass[conference, a4paper]{IEEEtran}

\usepackage{tikz}

\usepackage{dirtytalk}

\usepackage{xcolor}

\usepackage{amsmath}
\usepackage{amsthm}
\usepackage{amssymb}
\usepackage{graphicx} 
  	\graphicspath{{./figures_v2/}}
\usepackage{subfigure}
\usepackage{url}
\usepackage{booktabs} 
\usepackage{array} 
\usepackage{verbatim} 
\usepackage{caption}

\usepackage{enumitem}
\usepackage{bbm}

\usepackage[font=footnotesize]{caption}

\newtheorem{theorem}{Theorem}

\newtheorem{definition}{Definition}
\newtheorem{lemma}{Lemma}

\usepackage{amsthm}

\newtheorem{problem}{OP}

\newtheorem*{corollary}{Corollary}
\newtheorem*{intermediate}{Intermediate Step}

\usepackage{color}

\newcommand{\eq}[1]{Eq.~\eqref{#1}}
\usepackage{soul}
\setstcolor{red}

\newcommand{\oldtext}[1]{}

\newcommand{\myitem}[1]{\vspace{0.25\baselineskip}\noindent\textbf{#1}}
\newcommand{\secref}[1]{Section~\ref{#1}}

\usepackage{comment}
\usepackage{algorithm}
\usepackage{algorithmicx}
\usepackage{algpseudocode}

\def\BibTeX{{\rm B\kern-.05em{\sc i\kern-.025em b}\kern-.08em T\kern-.1667em\lower.7ex\hbox{E}\kern-.125emX}}

\begin{document}

\newfloat{subroutine}{htbp}{loa}
\floatname{subroutine}{Subroutine}
\title{The Order of Things: Position-Aware Network-friendly Recommendations in Long Viewing Sessions} 

\author{
\large{
Theodoros Giannakas\textsuperscript{1}, 
Thrasyvoulos Spyropoulos\textsuperscript{1},
and Pavlos Sermpezis\textsuperscript{2}
}\\
\normalsize
\textsuperscript{1}~EURECOM,~France, first.last@eurecom.fr\\
\textsuperscript{2}~FORTH,~Greece, sermpezis@ics.forth.gr 
}

\maketitle
\thispagestyle{plain}
\pagestyle{plain}
\begin{abstract}
Caching has recently attracted a lot of attention in the wireless communications community, as a means to cope with the increasing number of users consuming web content from mobile devices. Caching offers an opportunity for a win-win scenario: nearby content can improve the video streaming experience for the user, and free up valuable network resources for the operator. At the same time, recent works have shown that recommendations of popular content apps are responsible for a significant percentage of users requests. As a result, some very recent works have considered how to nudge recommendations to facilitate the network (e.g., increase cache hit rates). In this paper, we follow up on this line of work, and consider the problem of designing cache friendly recommendations for long viewing sessions; specifically, we attempt to answer two open questions in this context: (i) given that recommendation position affects user click rates, what is the impact on the performance of such network-friendly recommender solutions? (ii) can the resulting optimization problems be solved efficiently, when considering both sequences of dependent accesses (e.g., YouTube) and position preference? To this end, we propose a stochastic model that incorporates position-aware recommendations into a Markovian traversal model of the content catalog, and derive the average cost of a user session using absorbing Markov chain theory. We then formulate the optimization problem, and after a careful sequence of equivalent transformations show that it has a linear program equivalent and thus can be solved efficiently. Finally, we use a range of real datasets we collected to investigate the impact of position preference in recommendations on the proposed optimal algorithm. Our results suggest more than 30\% improvement with respect to state-of-the-art methods.

\end{abstract}

\section{Introduction}
\label{sec:intro}
Storing content close to wireless users is recognized as a promising method to (i) reduce the network cost to serve a request, and (ii) improve user experience (e.g., better playout quality). As a result, a number of studies suggest to install tiny caches (e.g., hard drives) at every small-cell or femto-node~\cite{femto}, bringing ideas from hierarchical caching~\cite{borst2010} into the wireless domain. 

%

Nevertheless, the rapidly growing catalog 
sizes, smaller sizes per cache (e.g., at femto-nodes or user devices) compared to traditional CDNs, and volatility of user demand when considering smaller populations, make the task of caching algorithms increasingly challenging~\cite{Paschos-misconceptions,ElayoubiRoberts15}. For example, installing say 1TB in every small cell in an ultra-dense network (already a pretty expensive investment) would still fit less $1\%$ or less of the content catalogue of even one provider (e.g., the Netflix catalogue is reportedly in the order of few PBs). Things are even more stringent for UE-side caching solutions \cite{ji2016fundamental}, where it is reported that up to 10-20 files could be pre-fetched realistically~\cite{karamshuk2016take}.

To overcome such challenges, a radical approach has been recently proposed~\cite{chatzieleftheriou2017technical,cache-centric-video-recommendation,giannakas2018show,content-recommendation-swarming,sermpezis2018soft,kastanakis-cabaret-mecomm}, based on the observation that user demand is increasingly driven today by recommendation systems of popular applications (e.g., Netflix, YouTube). Instead of simply recommending \emph{interesting content}, recommendations could instead be ``nudged'' towards \emph{interesting content with low access cost} (e.g., locally cached)~\cite{chatzieleftheriou2017technical,kastanakis-cabaret-mecomm}:
the recommendation quality remains unaltered, and the new content might in fact become accessible at better quality (e.g., HD).
This idea is appealing, potentially presenting a win-win situation for all involved parties. 

Nevertheless, due to the very recent research interest in the topic, a number of key questions remain unanswered. First, it has been shown that the users have the tendency to click on recommended contents (or products in the case of e-commerce) according to the position they find them, e.g., contents higher up in the recommendation list~\cite{what-should-you-cache-nossdav,RecImpact-IMC10}. However, several of the aforementioned studies tend to ignore this aspect~\cite{giannakas2018show,sermpezis2018soft,content-recommendation-swarming} in their analysis, assuming that an equally good recommendation will be clicked equally frequently, regardless of the position in the application GUI that it appears. The work in~\cite{chatzieleftheriou2017technical}, while taking into account the ranking of the recommendations in the modeling and their proposed algorithm, in the simulation section they assume that the boosting of the items is equal. So an interesting question arising then is: \emph{Does the performance of network-friendly recommendation schemes improves, deteriorates, or is unaffected by such position preference?}

A second important question has to do with the computational complexity of optimizing network-friendly recommendations. In settings where each user requests one content (or equivalently requests many contents in an I.I.D. manner), the caching-side of the problem~\cite{sermpezis2018soft} or the recommendation-side of the problem
, can be efficiently approximated. However, the joint caching and recommendation problem is NP complete~\cite{chatzieleftheriou2017technical}, without any known approximation guarantees or optimal decompositions~\cite{chatzieleftheriou2017technical}. Things get worse, when one considers a user accessing multiple contents during a session in a \emph{structured} manner, due to the inherent memory this system has (the content recommended and/or accessed at step $n$ has an impact beyond step $n+1$). \emph{Even without position preference}, the problem of network-friendly recommendations for long (markovian) sequences of content accessed seems to be hard (non-convex)~\cite{giannakas2018show}. A second question of interest then is: \emph{Can the problem of network-friendly recommendations even be solved efficiently, in a context where there is both position preference and dependence in consecutive content requests?}

To this end, in this paper we make the following contributions towards answering the above questions:

\noindent \myitem{(i) Sequential request analysis based on absorbing Markov chain theory.}
We propose an analytical framework based on absorbing Markov chain theory, to model a user accessing a sequence 
of contents, driven by a recommender (Sections~\ref{sec:problem_setup} and~\ref{sec:problem_formulation}).
The sequential request model with preference to top recommendations better fits real users behavior in a number of popular applications
(e.g. YouTube, Vimeo, Spotify) compared to Independent Reference Models (IRM) used in previous work~\cite{sermpezis2018soft} and/or models neglecting the position of recommendations~\cite{giannakas2018show,sermpezis2018soft}.

\noindent \myitem{(ii) Optimal solution.}
We formulate a generic optimization problem for high quality but network-friendly recommendations. 
While the original problem in non-convex (similarly to previous formulations~\cite{giannakas2018show}), we prove an equivalent convex one through a sequence of transformations, which allows to solve the original problem efficiently (Section~\ref{sec:optimization_methodology}). 

\noindent \myitem{(iii) Real data analysis and performance evaluation.}
%
%
%
 We validate our algorithms using existing and collected datasets from different content catalogs (e.g., YouTube, Movielens), and demonstrate performance improvements up to 35\% compared to a state-of-the-art method, and 60\% compared to a greedy cache-friendly recommender (in terms of relative gain), for a scenario with 90\% of the original recommendation quality (Section~\ref{sec:sims}). Our findings reveal that the more skewed the preference towards top positions of recommendations is, the higher the gains of network-friendly recommendation schemes can be. 

Finally, we discuss related work in \secref{sec:related} and conclude our paper in \secref{sec:conclusions}.



\section{Problem Setup}
\label{sec:problem_setup}
\subsection{Recommendation-driven Content Consumption.} We consider a user that consumes one or more contents during a session, drawn from a catalogue $\mathcal{K}$ of cardinality $K$. It is reported that YouTube users spend on average around 40 minutes at the service, viewing several related videos~\cite{businessYoutubeSessions}. After each viewing, a user is offered some recommended items that she might follow or not, according to the  model below. 
\begin{definition}[Recommendation-Driven Requests]\label{def:requests} After a user consumes a content, $N$ contents are recommended to her (these might differ between users).
\begin{itemize}[leftmargin=*,noitemsep,topsep=0pt]
  \item \emph{with probability $1-\alpha$} ($\alpha\in[0,1]$) she ignores the recommendations, and picks a content $j$ (e.g., through a search bar) with probability $p_{j} \in (0,1)$, $\mathbf{p}_{0} = [p_{1}, p_{2}, ..., p_{K}]^{T}$.
  \item \emph{with probability} $\alpha$ she follows one of the $N$ recommendations. 
 \item each of the $N$ recommended contents is placed in one of $N$ possible slots/positions in the application GUI; if she \emph{does} follow recommendation, the \emph{conditional} probability to pick the item in position $i$ is $v_i$, where $\sum_i v_i = 1$.
\end{itemize}
\end{definition}
We assume the probabilities $p_j$ capture long-term user behavior (beyond one session), and possibly the impact of the baseline recommender. W.l.o.g. we also assume $\mathbf{p}_{0}$ governs the first content accessed, when a user starts a session. 
This model captures a number of everyday scenarios (e.g., watching clips on YouTube, personalized radio, etc). 

The last point in the definition is a key differentiator of this work, compared to some previous ones on the topic~\cite{giannakas2018show},~\cite{chatzieleftheriou2017technical},~\cite{content-recommendation-swarming}. A variety of recent studies~\cite{what-should-you-cache-nossdav,RecImpact-IMC10} has shown that the web-users have the tendency to click on contents (or products in the case of e-commerce) according to the position they find them. For example, in the PC interface of YouTube, they show a preference for the contents that are higher in the list of the recommended items. Hence, the probability of picking content in position 1 ($v_1$), might be quite higher than the probability to pick the content in position $N$ ($v_N$)\footnote{In fact, a Zipf-like relation has been observed~\cite{RecImpact-IMC10}.}. In contrast,~\cite{giannakas2018show,chatzieleftheriou2017technical,content-recommendation-swarming} explicitly or implicitly assume that $v_i = \frac{1}{N}, \forall i$.  

\emph{Remark - Position Entropy}: A key goal of this paper is to understand the additional impact of position preference on the achievable gains of network-friendly recommendations. A natural way to capture position preference is with the \emph{entropy} of the probability mass function $\mathbf{v} = [v_1,v_2,...,v_N]$, namely
\begin{equation}
   H_{\mathbf{v}} =  H(v_1,..,v_N) = -\sum_{n=1}^{N} v_n \cdot \log(v_n).
\end{equation}
The original case of no position preference, corresponds to a uniformly distributed $\mathbf{v}$, which is well known to have maximum entropy. Any position preference will lead to lower entropy, with the extreme case of a ``1-hot vector'' (i.e., only one $v_i =1$) having zero entropy. 

\myitem{Content Retrieval Cost.} We assume that fetching content $i$ is associated with a generic cost $c_{i}\in\mathbb{R}$, $\mathbf{c} = [c_{1}, c_{2},...,c_{K}]^{T}$, which is known to the content provider, and might depend on access latency, congestion overhead, or even monetary cost. 

\noindent\emph{Maximizing cache hits:} Can be captured by setting $c_i = 1$ for all cached content and to $c_i = 0$, for non-cached content. 

\noindent\emph{Hierachical caching:} Can be captured by letting $c_i$ take values out of $n$ possible ones, corresponding to $n$ cache layers: higher values correspond to layers farther from the user~\cite{poularakis2014toc,borst2010}. 


\subsection{Baseline Recommendations.} 

Recommendation systems (RS) are an active area of research, with state-of-the-art RS using collaborative filtering~\cite{sarwar2001item},
and recently deep neural networks~\cite{covington2016deep}. 
%
For simplicity, we assume that the baseline RS works as follows:
\begin{definition}[Baseline Recommendations and Matrix $\mathbf{U}$]\label{def:baseline-RS} ~

\noindent (i) For every pair of contents $i,j \in \mathcal{K}$ a score $u_{ij}\in[0,1]$ is calculated, using a state-of-the-art method. Note that these scores can be personalized, and differ between users.
\footnote{$u_{ij}$ could correspond to the \emph{cosine similarity} between content $i$ and $j$, in a collaborative filterting system~\cite{sarwar2001item}, or simply take values either $1$ (for a small number of related files) and $0$ (for unrelated ones). These scores might also depend on user preferences and past history of that user, as is often the case when users are logged into the app.}

\noindent(ii) After a user has just consumed content $i$, the RS recommends contents according to these $u_{ij}$ values (e.g.,  the $N$ contents $j$ with the highest $u_{ij}$ value~\cite{RecImpact-IMC10,covington2016deep}.\footnote{$N$ depends on the scenario. E.g., in YouTube $N=2,..,5$ in its mobile app, and $N=20$ in its website version.}
\end{definition}

\subsection{Network-friendly Recommendations.} 

Our goal is to depart from the baseline recommendations (Def.~\ref{def:baseline-RS}) that are based only on $\mathbf{U}$, and let them consider the access costs $\mathbf{c}$ as well. We define recommendation decisions as follows.

\begin{definition}[Control Variables $\mathbf{R}^{1},..,\mathbf{R}^{N}$]\label{def:control-variables}
Let $r^{n}_{ij} \in [0,1]$ denote the probability that content $j$ is recommended after a user watches content $i$ in the position $n$ of the list. For the $n$-th position in the recommendation list, these probabilities define a matrix $K \times K$ recommendation matrix, which we call $\mathbf{R}^{n}$. 
\end{definition}


Defining recommendations as probabilities provides us more flexibility, as it allows to not always show a user the same contents (after consuming some content $i$).
For example, 
assume $K=4$ total files, a user just watched item 1, and $N=2$ items must be recommended. Let the first row of the matrix $\mathbf{R}^{1}$ be
$\mathbf{r}_{1}^{1} = [0, 1, 0, 0]$ and that of $\mathbf{R}^2$ be $\mathbf{r}_{1}^{2} = [0, 0, 0.5, 0.5]$. In practice, this means that in position 1 the user will always see content 2 being recommended (after consuming content 1), and the recommendation for position 2 will half the time be for content 3 and half for content 4.

Our objective is to choose what to recommend in which position, i.e., choose $\mathbf{R}^{1},.. \mathbf{R}^{N}$, to minimize the average content access cost. However, we still need to ensure that the user remains generally happy with the quality of recommendations and does not abandon the streaming session.

\myitem{Recommendation Quality Constraint.}

Let $r^{n(B)}_{ij}$ denote the baseline recommendations of Def~.\ref{def:baseline-RS}. We can define the recommendation quality of this baseline recommender for content $i$, $q_i^{max}$ as follows
\begin{align} \label{eq:qmax}
q^{max}_{i} = \sum_{j=1}^{K} \sum_{n=1}^{N} v_n \cdot r^{(n)(B)}_{ij} \cdot u_{ij}.
\end{align}
This quantity will act as another figure of merit for other (network-friendly) RS.

\begin{definition}[Quality of Network-Friendly Recommendations]\label{def:quality-rec}
Any other (network-friendly) RS that differs from the baseline recommendations $r^{B}_{ij}$ can be assessed in terms of its recommendation quality  $q \in [0,1]$ with the constraint:
\begin{equation}\label{eq:quality}
\sum_{j=1}^{K} \sum_{n=1}^{N} v_n \cdot r^{n}_{ij} \cdot u_{ij} \ge q \cdot q^{max}_{i}, \forall i \in \mathcal{K}.
\end{equation}
where $q^{max}_{i}$ is the quantity defined in Eq.(\ref{eq:qmax}).
\end{definition}
This equation weighs each recommendation with: (a) its quality $u_{ij}$, and (b) the importance of the position $n$ it appears at, $v_n$. Note however that this constraint is not a restrictive choice. One could conceive a more ``aggresive'' recommender that removes the weight $v_n$ from the left-hand side. In fact, our framework can handle any quality constraint(s) that are convex in $r^{n}_{ij}$. 

Based on the above discussion, a network-friendly recommendation could favor at each step contents $j$ (i.e., give high $r^{n}_{ij}$ values) that have low access cost $c_j$ but also are interesting to the user (i.e., have high $u_{ij}$ value). 
However, as we show in later sections%
, such a greedy approach is suboptimal, as the impact of $r_{ij}^{n}$ goes beyond the content $j$ accessed next
, affecting the entire \emph{sample path} of subsequent contents in that session. 
The example in Fig.~\ref{fig:example-comparison-rec} depicts such a scenario: after content $3$, instead of recommending content $6$ (related value $u_{36}=1$) content $4$ is recommended ($u_{34}=0.8$), because $4$ is more related to cached contents ($9$ and $7$) that can be recommended later (whereas $6$ is related to the non-cached contents $5$ and $4$).\footnote{The reason is that many contents $j$ will have high enough relevance $u_{ij}$ to the original content $i$, and are thus interchangeable~\cite{RecImpact-IMC10}}




\begin{figure}
\centering
\subfigure{\includegraphics[width=0.8\columnwidth]{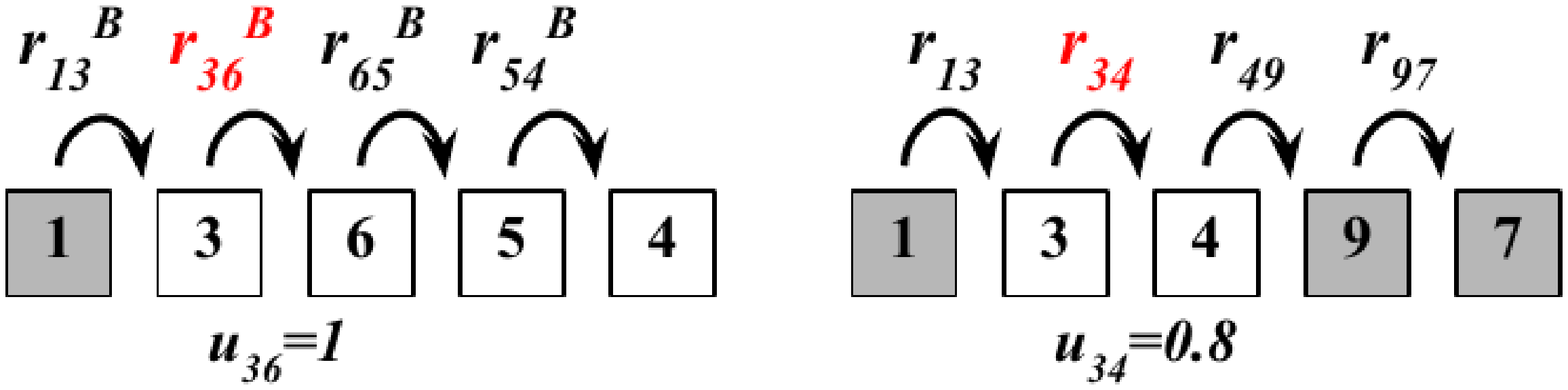}}
\vspace{-2.5mm}
\caption{Comparison of baseline (left) and network-friendly (right) recommenders. Gray and white boxes denote cached and non-cached contents, respectively. Recommending after content $3$ a slightly less similar content (i.e., content $4$ instead of $6$), leads to lower access cost in the long term.}
\label{fig:example-comparison-rec}
\end{figure}

\myitem{Remark on Recommendation Personalization.} As hinted at earlier, content utilities $u_{ij}$ and recommendations $r^{n}_{ij}$ can be user-specific (e.g. $u_{ij}^{u}$ for user $u$), since different users might have different access patterns that can be leveraged. Nevertheless, to avoid notation clutter we do not use superscript $u$ in the remainder of the paper, and will assume that these quantities and the respective optimization algorithm is done per user.

\myitem{Remark on Recommendation Quality.} Cache-friendly recommendations might also improve user QoE, in addition to network cost, a ``win-win'' situation. Today's RS, measure their performance (QoR) without taking into account where the recommended content is stored.
Assuming two contents equally interesting to the user where the one \emph{is} stored locally while the other \emph{is not}; it is obvious that the cached one could be streamed in much better quality (e.g., HD, so higher QoS), thus leading to $q > 1$. Hence, more sophisticated QoE (= QoR + QoS) metrics \emph{could} combine these effects: e.g., a content's effective utility $\hat{u}_{ij} = f(u_{ij},c_j)$ that increases if $j$ is highly related to $i$ but also if it is locally cached (i.e., $c_j$ is low). Such a metric could be immediately integrated into our framework, simply be replacing $u$ with $\hat{u}$.




Table~\ref{table:notation} summarizes some important notation. Vectors and matrices are denoted with bold symbols.

\vspace{5pt}
\begin{table}
\centering
\caption{Important Notation}\label{table:notation}
\begin{small}
\begin{tabular}{|l|l|}
\hline
{$\alpha$}			&{Prob. the user follows recommendations}\\
\hline
{$r_{ij}^{n}$}			&{Prob. to recommend $j$ after $i$ at position $n$}\\
\hline
{$q_{i}^{max}$}		&{Maximum baseline quality of content $i$}\\
\hline
{$q$}		        &{Percentage of original quality}\\
\hline
{$\mathbf{p}_0$}	&{Baseline popularity of contents}\\
\hline
{$u_{ij}$}			&{Similarity scores content pairs $\{i,j\}$, included in $\mathbf{U}$}\\
\hline
{$v_{n}$}			&{Click prob. of recommendation at the position $n$}\\
\hline
{$c_i$}			    &{Access cost for content $i$}\\
\hline
{$\mathcal{K}$}		&{Content catalogue (of cardinality $K$)}\\
\hline
{$N$}				&{Number of recommendations}\\
\hline
\end{tabular}
\end{small}
\end{table}

\section{Average Session Cost}
\label{sec:problem_formulation}

Having defined the content access model, our first step towards ``optimizing'' the (network-friendly) recommendations, is to better understand what we are trying to optimize. To this end, in this section we derive the expected content access cost for a typical user session, as a function of recommendation variables $r_{ij}^{n}$. This will serve as the \emph{objective} of our problem.  (Section~\ref{sec:optimization_methodology}).




\begin{definition}\label{lemma:process-as-a-markov-chain}
Let $S = \{i_{1},i_{2},\dots, i_{s} \}, i_{n} \in \mathcal{K}$ be a sequence of contents accessed by a user according to Def.~\ref{def:requests} during a viewing session. Then $S$ is a discrete-time Markov process with transition matrix
\begin{align}\label{eq:MC}
\mathbf{P} = \alpha \sum_{n=1}^{N} v_n  \mathbf{R}^{n} + (1-\alpha) \mathbf{1} \cdot \mathbf{p}_{0}^{T},
\end{align}
where $\mathbf{1} = [1,1,...,1]^{T}$ is a column vector of all 1s. 
\end{definition}

When the user has just consumed content $i$, then she might next consume content $j$ if all the following occur: she decides to follow a recommendation (probability $\alpha$ according to Def.~\ref{def:requests}), $j$ appears in the position $n$ (probability $r^{n}_{ij}$), and she picks the content at the $n$-th position (probability $v_n$). These probabilities are by definition independent, hence the probability of these three events is their product, $\alpha \cdot v_n \cdot r^{n}_{ij}$. Note that the user might consume $j$, if she finds it in positions other than $n$ (for example in position $m$) and will then click it with $v_m$. Moreover, the user might also consume $j$ after $i$, if she ignores the recommendations (with probability $1-\alpha$ according to Def.~\ref{def:requests}) and picks content $j$ from the entire catalog (with probability $p_j$). Putting all these together gives the transition probability from $i$ to $j$, $Pr\{i \to j\} = \alpha \cdot \sum_{n=1}^{N} v_n \cdot r^{n}_{ij} + (1-\alpha)\cdot p_j$, which written in matrix notation gives \eq{eq:MC}. 

\begin{lemma}[Content Access as Renewal-Reward]\label{lemma:RR}
A content access sequence $S = \{S_{R}^{1},S_{R}^{2},\dots\}$ defines a renewal process, with subsequences $S_{R}$, where the user follows recommended content, each ending with a jump outside of the recommender. The cost $c_{i}$ incurred at each state is the reward.
\end{lemma}

It is easy to see that whenever a user makes a jump outside of the recommendations (w.p.  $1-\alpha$), the process renews to state $\mathbf{p_0}$. An example can be found in Fig.~\ref{fig:example2}.
\begin{figure}
\centering
\subfigure{\includegraphics[width=0.6\columnwidth]{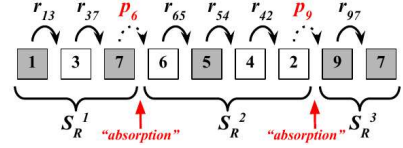}}
\vspace{-2mm}
\caption{Example of a multi-content session. Gray and white boxes denote cached and non-cached contents, respectively. A user follows recommendations (continuous arrows) or ignores them (dotted arrows).}
\label{fig:example2}
\end{figure}
%
To derive the mean access cost, we employ Lemma~\ref{lemma:RR} and the framework of Absorbing Markov Chains (AMC) \cite{grinstead2012introduction}: a user is in \emph{transient} states while she is following recommendations; and she gets \emph{absorbed} as soon as a jump outside of recommendations occurs, as shown in Fig.~\ref{fig:example2}. Hence, during a content access sequence, recommendations affect the user's choices (and related costs) only during the transient states.  
\begin{lemma}[Recommendation-Driven Cost] \label{lemma:rec-cost}
The content access cost $C(S_{R})$ during a renewal cycle $S_{R}$ is given by
\begin{equation}\label{eq:cost-cycle}
E[C(S_R)] = \mathbf{p}_0^{T} \cdot \mathbf{G} \cdot \mathbf{c},
\end{equation}
and the expected length of such a cycle is
\begin{equation}\label{eq:cycle-length}
|S_{R}| = \mathbf{p}_0^{T} \cdot \mathbf{G} \cdot \mathbf{1} = \frac{1}{1-\alpha},
\end{equation}
where $\mathbf{G} = \left(\mathbf{I}-\alpha \cdot \sum_{n=1}^{N} \cdot \mathbf{R}^{n}\right)^{-1}$ is the \emph{Fundamental Matrix} of an AMC with $K$ transient states and $1$ absorbing state, corresponding to a jump outside recommendations. 
\end{lemma}
\begin{proof}
Let a user start a sub-sequence by retrieving content $i$. The expected number of retrievals of content $j$ (or, number of times visiting state $j$) until the end of the sub-sequence is given by $g_{ij}$, where $g_{ij}$ is the ($i$-row,$j$-column) element of the \textit{fundamental matrix} $\mathbf{G}$ of the AMC~\cite{grinstead2012introduction}.

The fundamental matrix is defined as 
\begin{equation}\label{eq:fundamental-matrix-definition}
\mathbf{G} = \textstyle\sum_{n=0}^{\infty} \mathbf{Q}^{n} = (\mathbf{I}-\mathbf{Q})^{-1}
\end{equation}
where $\mathbf{Q}$ the matrix with the transition probabilities $q_{ij}$ between the transient states of the AMC ($i,j\in\mathcal{K}$). Following the same arguments as in Def.~\ref{lemma:process-as-a-markov-chain}, we get that $q_{ij} =  \alpha \cdot \sum_{n=1}^{N} v_n \cdot r^{n}_{ij}$, or, in a matrix format 
$\mathbf{Q} = \alpha \cdot \sum_{n=1}^{N} v_n\cdot \mathbf{R}^{n}.$
Substituting this into \eq{eq:fundamental-matrix-definition} gives the expression for $\mathbf{G}$ that appears in Lemma~\ref{lemma:rec-cost}.
Now, the cost of retrieving a content $j$ is $c_{j}$. Since each content $j$ is retrieved on average $g_{ij}$ times during a sub-sequence that starts from $i$, the total cost is given by 
\begin{equation}
E[C(S_{R})~|~ i] = \textstyle\sum_{j\in\mathcal{K}} g_{ij}\cdot c_{j}
\end{equation}

The probability that a sub-sequence starts at content $i$ is equal for all sub-sessions and is given by $p_{i}$. Thus, taking the expectation over all the possible initial states $i$, gives

\begin{equation}
E[C(S_{R})] = \sum_{i\in\mathcal{K}} E[C(S_{R})~|~ i] \cdot p_{i} =  \sum_{i\in\mathcal{K}} \sum_{j\in\mathcal{K}} g_{ij}\cdot c_{j} \cdot p_{i}
\end{equation}

Expressing the above summation as the product of the vectors $\mathbf{p}_{0}$ and $\mathbf{c}$, and the matrix $\mathbf{G}$, gives \eq{eq:cost-cycle}.

Similarly, if $g_{ij}$ is the amount of time spent on state $j$ before absorption, starting from state $i$, then $\sum_{j} g_{ij}$ must be equal to the total time spent at \emph{any state} before absorption
. Weighing this with the probability $p_i$ of starting at each state $i$, gives the expected time to absorption, which is the expected duration of a sub-sequence $E[|S_R|] = \sum_{i} p_{i} \cdot \sum_{j} g_{ij}$. Writing this in matrix notation, gives the first part of Eq.(\ref{eq:cycle-length}). 

However, observe that the probability of absorption at any state $i$ is equal to $1-\alpha$, independent of $i$. Hence, the number of steps till absorption is a \emph{geometric} random variable with parameter $1-\alpha$, and thus the mean time (i.e., number of steps) to absorption is $\frac{1}{1-\alpha}$.
\end{proof}

The following Theorem, which gives the expected retrieval cost for a user session, follows immediately from Lemmas~\ref{lemma:RR},~\ref{lemma:rec-cost}, and the Renewal-Reward theorem~\cite{mor2013}

\begin{theorem}\label{thm:total-expected-cost}
The expected retrieval cost per content, for a user session S, given a recommendation matrix $\mathbf{R}$ is
\begin{equation}\label{eq:cost-rate}
E[C(S)~|~\mathbf{R}^{1},.,\mathbf{R}^{N}] = \frac{\mathbf{p}_0^{T} \left(\mathbf{I}-\alpha \sum_{n=1}^{N} v_n  \mathbf{R}^{n}\right)^{-1}  \mathbf{c}}{\frac{1}{1-\alpha}}
\end{equation}
\end{theorem}
\section{Optimization Problem and Methodology}
\label{sec:optimization_methodology}
In this section, we use the results of the previous section to formulate the problem of minimizing the expected access cost until absorption under a set of modeling constraints.

\subsection{The Problem and its Constraints}
\begin{problem} [Nonconvex formulation]
\begin{small}
\label{problem:basis}
\begin{align}
\underset{\mathbf{R}^{1},.., \mathbf{R}^{N}}{\textnormal{minimize}} \; & {\mathbf{p}_{0}^{T} \cdot (\mathbf{I} - \alpha \cdot \sum_{n=1}^{N} v_n \cdot \mathbf{R}^{n})^{-1}}\cdot \mathbf{c}\label{eq:objective}\\
\textnormal{subject to} \quad
 & \sum_{j=1}^{K}\sum_{n=1}^{N} v_{n} \cdot r_{ij}^{n} \cdot u_{ij} \geq q\cdot q_{i}^{max},~\forall~i \in \mathcal{K} \label{quality-con1}\\ 
& \sum_{j = 1}^{K} r^{n}_{ij} = 1,~\forall~i \in \mathcal{K}~\textnormal{and}~n=1,...,N \label{affine-con1} \\
& \sum_{n = 1}^{N} r^{n}_{ij} \le 1,~\forall~\{i,j\} \in \mathcal{K} \label{constraint-implementable} \\
& 0 \le r^{n}_{ij} \le 1 \; (i \ne j), \;\; r^{n}_{ii} = 0~\forall~i,~n. \label{box-con}
\end{align}
\end{small}
\end{problem}

The constraint in Eq.(\ref{quality-con1}), is responsible for keeping the quality of the recommendations above a pre-specified (and given) threshold. The pair of constraints in Eqs.(\ref{affine-con1},\ref{box-con}), defines a probability simplex for every row of all the $\mathbf{R}^{n}$ matrices. Note that we also prohibit self-recommendations ($r_{ii}^{n} = 0~\forall~i$ and $n$) (see Eq.(\ref{box-con})).
Importantly, Eq.(\ref{constraint-implementable}) is necessary in the position-aware setup, to ensure that the same content will not be recommended in two different positions.
As an example assume that $\mathbf{r}_{1}^{1} = [0, 1, 0, 0]$ and $\mathbf{r}_{1}^{2} = [0, 0.2, 0.3, 0.5]$, in that case we clearly see that content 2 would always be shown in position 1 (after watching content 1), but 20\% of those times it would be shown in position 2 as well. Hence, Eq.(\ref{constraint-implementable}) ensures that such decision vectors would be \emph{infeasible}.

Evidently, our feasible space consists of either linear (equalities or inequalities) or box constraints with respect to the decision variables $r^{n}_{ij}$. However, the objective is non-convex in general.



\begin{lemma}\label{lemma:initial-nonconvex}
The problem described in \textbf{OP~\ref{problem:basis}} is nonconvex.
\end{lemma}

\begin{proof}
The problem \textbf{OP~\ref{problem:basis}} comprises $N\cdot K^{2}$ variables $r^{n}_{ij}$, and a set of $K^{2}\cdot(N+2) + K$ linear (equality and inequality) constraints, thus the feasible solution space is convex. However, assume w.l.o.g that $\mathbf{p}_0 = \mathbf{c} = \mathbf{w}$, $N=1$, and $v_1 = 1$; the objective now becomes $f(\mathbf{R})=\mathbf{w}^{T}(\mathbf{I}-\alpha \cdot \mathbf{R})^{-1}\mathbf{w}$. 
Unless $\mathbf{R}$ is symmetric positive semi-definite (PSD), $f(\mathbf{R})$ is non-convex~\cite{boyd2004convex}. Forcing $\mathbf{R}$ to be symmetric would \emph{require additional} constraints that lead to suboptimal solutions of this problem~\cite{ermon2014designing}.
%
%
Therefore, our objective as is, is nonconvex and there are no exact methods that can solve it in polynomial time.
\end{proof}

\subsection{The Journey to Optimality}\label{sec:journey-optimality}
In addition to non-convexity, a key difficulty in solving \textbf{OP~\ref{problem:basis}} is the inverse matrix in the objective. Any gradient-based algorithm would require a matrix inversion at each gradient step (an operation of complexity $\mathcal{O}(K^3)$. 
To circumvent this, we introduce $K$ auxiliary variables $\mathbf{z}^{T}$, for which we will demand $\mathbf{z}^{T} = \mathbf{p}_0^T \cdot(\mathbf{I} -\alpha \cdot \sum_{n=1}^{N} v_n \mathbf{R}^{n})^{-1}$. This introduces $K$ new equality constraints, leading to the following equivalent problem.\footnote{Two problems are equivalent if the solution of the one, can be uniquely obtained through the solution of the other~\cite{boyd2004convex}; introducing auxiliary variables preserves the property.}

\begin{intermediate}[Equivalent formulation]\label{problem:eq-1}
\begin{small}
\begin{align}
\underset{\mathbf{z},~\mathbf{R}^{1},..,\mathbf{R}^{N}}
{\textnormal{minimize}} \; & 
  \mathbf{c}^T \cdot \mathbf{z},\label{eq:objective-interm1}\\
\textnormal{subject to} \quad &  \mathbf{z}^T - \alpha \cdot \mathbf{z}^T\cdot \sum_{n=1}^{N} v_n \cdot \mathbf{R}^{n} = \mathbf{p}_0^T  \label{stationarity-con} \\
& \textnormal{Eqs}.(\ref{quality-con1},\ref{affine-con1},\ref{constraint-implementable},\ref{box-con}) 
\end{align}
\end{small}
\end{intermediate}
The new objective is now convex (in fact, linear) in the new variable ($\mathbf{z}$). However, as the set of constraints Eq.(\ref{stationarity-con}) are all \emph{quadratic equalities}, the problem remains nonconvex. 
The above formulation falls under the umbrella of non-convex QCQP (Quadratically Constrained Quadratic Program), where it is common to perform a convex relaxation of the quadratic constraints, and then solve an approximate convex problem (e.g., SDP or Spectral relaxation, see~\cite{park2017general} for more details). The problem can also be seen as \emph{bi-convex} in variables $\mathbf{R}^{n}$ and $\mathbf{z}$, respectively. Alternating Direction Method of Multipliers (ADMM) can be applied to such problems, iteratively solving convex subproblems~\cite{boyd2011distributed,giannakas2018show}. Nevertheless, none of these methods provides any optimality guarantees, and even convergence for non-convex ADMM is an open research topic.






To further deal with this additional complication, we define another set of variables as $f^{n}_{ij} =  z_{i} \cdot r^{n}_{ij}$. Since the $j$-th element of the $n$-th vector $\mathbf{z}^T \cdot \mathbf{R}^{n}$ can be written as $ \sum_{i} z_{i} \cdot r^{n}_{ij}$, we can write now $  \mathbf{z}^T \cdot \mathbf{R}^{n} = \mathbf{1}^T \cdot \mathbf{F}^{n}$, and the new variables are $\mathbf{z}$ and $\mathbf{F}^{1},..,\mathbf{F}^{N}$, which are a $K \times 1$ vector, and  $N$ $K \times K$ matrices respectively.

This new transformation leads to the following problem.
\begin{problem}[LP formulation]\label{problem:basis-LP}
\label{eq:objective-LP}
\begin{small}
\begin{align}
\underset{\mathbf{z},~\mathbf{F}^{1},..,\mathbf{F}^{N}}{\textnormal{minimize}} \; \; \; \; & 
  \mathbf{c}^T \cdot \mathbf{z}, \label{new-obj2}\\
\textnormal{subject to} \quad & \sum_{j =1}^{K} \sum_{n=1}^{N} v_{n} \cdot f^{n}_{ij} \cdot u_{ij} \geq z_i \cdot q \cdot q_i^{max},~\forall~i~\in \mathcal{K}
\label{quality-con-LP}\\
& \sum_{j = 1}^{K} f^{n}_{ij} = z_i , ~\forall~i~\in \mathcal{K}~\textnormal{and}~n=1,..,N \label{affine-con-LP}\\
& \sum_{n = 1}^{N} f^{n}_{ij} \le z_i,~\forall~\{i,~j\} \in \mathcal{K} \label{implementable-lp} \\
& f^{n}_{ij} \ge 0 \; (i \neq j), \;\; f^{n}_{ii} = 0, \forall~i,j~\in~\mathcal{K} \label{f-positive-con-LP}\\
&	z_j - \alpha \cdot \sum_{n=1}^{N} v_n \cdot \sum_{i}^{K}f^{n}_{ij} = p_{j}, ~\forall j \in \mathcal{K} \label{zf-relaxation}
\end{align}
\end{small}
\end{problem}

\begin{lemma}\label{lemma:bijection}
The change of variables $f_{ij}^{n} =  z_{i} \cdot r_{ij}^{n}$, is a bijection (one-to-one mapping) between $(z_i,r_{ij}^{n})$ and $(z_i,f_{ij}^{n})$. 
\end{lemma}

\begin{proof}
This follows immediately, as we can readily obtain $r_{ij}^{n} = \frac{f_{ij}^{n}}{z_{i}}$ from $\{z_i, r_{ij}^{n} \}$. Note that, since $z_{j} = \sum_{i} f_{ij}^{n} + p_{j}$, and $p_{i}~\in~(0,1)~\forall~i$, i.e. nonzero (see Def.~\ref{def:requests}), this forces $\mathbf{z} > \mathbf{0}$ and thus $r_{ij}^{n}$ are always uniquely defined.
\end{proof}

\begin{corollary}
\textbf{OP~\ref{problem:basis}} can be solved efficiently as an LP.
\end{corollary}

\begin{proof}
Equivalency due to Lemma~\ref{lemma:bijection}.
\end{proof}

We have therefore transformed the nonconvex \textbf{OP~\ref{problem:basis}} to a convex (LP) one \textbf{OP~\ref{problem:basis-LP}}, and can now solve it optimally.


\subsection{A Myopic Approach}\label{sec:greedy}
A natural way to tackle the \textbf{OP~\ref{problem:basis}} is to try minimizing the cost of content retrieval in a single-content session (i.e., only one transition in the Markov chain). This is equivalent to minimizing the scalar quantity
\begin{equation}
\mathbf{p}_0^{T} \cdot(\alpha \cdot \sum_{n=1}^{N} v_n \cdot \mathbf{R}^{n}+ (1-\alpha) \cdot \mathbf{1}^{T} \cdot \mathbf{p}_{0}) \cdot \mathbf{c}
\end{equation}
 Ignoring the terms that do not depend on the control variables $\mathbf{R}^{n}$, yields the following.
\begin{problem} [Greedy Aware Recommendations]
\label{problem:single}
\begin{small}
\begin{align} 
\underset{\mathbf{R}^{1},..,\mathbf{R}^{N}}{\textnormal{minimize}}~~ \; & 
\mathbf{p}_0^{T} \cdot \Big(\sum_{n=1}^{N} v_n \cdot \mathbf{R}^{n}\Big) \cdot \mathbf{c},\label{objective-single}\\
\textnormal{subject to} \quad & \textnormal{Eqs}.(\ref{quality-con1},\ref{affine-con1},\ref{constraint-implementable},\ref{box-con}) \label{constraint-set-myopic}
\end{align}
\end{small}
\end{problem}
Unlike the multi-step problem, this is already an LP, and can be solved directly without the earlier transformation steps. 
This solution of \textbf{OP~\ref{problem:single}} will serve in the upcoming results section as a baseline approach, to solving the hard basis problem \textbf{OP~\ref{problem:basis}}.
Interestingly, 
the solution of \textbf{OP~\ref{problem:single}} (we will call \emph{Greedy} from now), resembles the policies proposed in \cite{chatzieleftheriou2017technical,cache-centric-video-recommendation}. Although the algorithm of~\cite{chatzieleftheriou2017technical} targets a different context, i.e., the \emph{joint} caching and single access content recommendation, the Greedy algorithm could be interpreted as applying the recommendation part of~\cite{chatzieleftheriou2017technical} for each user, along with a continuous relaxation of the control (recommendation) variables. In doing so, the recommendation problem is simply an LP of the type of Eq.(\ref{objective-single}), when the recommendations are allowed to be probabilistic. Due to this relaxation, the greedy algorithm is an upper bound for~\cite{chatzieleftheriou2017technical}, looking at the recommendation problem \emph{only}.
\section{Validation Results}
\label{sec:sims}
\subsection{Warm Up}
In this section we evaluate the performance of the proposed algorithm and provide insights regarding the behavior of the network-friendly recommendations schemes. For a realistic evaluation, we use three collected datasets from video/audio services. Before diving into the details, we need to state the following

\noindent \emph{Performance metric}: \emph{Cache Hit Rate (CHR)}, as computed by the objective of Eq.(\ref{eq:cost-rate}), here we will minimize the cache miss.

\noindent \emph{Relative Gain}: computed as $\frac{CHR_{(\text{proposed})} - CHR_{(\text{baseline})}}{CHR_{(\text{baseline})}}\cdot 100 \%$.

\noindent\emph{$\mathbf{p}_0$}: drawn from Zipf~\cite{adamic2002zipf} of parameter $s$.

\noindent\emph{$\mathbf{v}$}: drawn from Zipf~\cite{RecImpact-IMC10} of parameter $\beta$ 

\noindent\emph{$\alpha$}: will vary from 0.7 to 0.8

\noindent \emph{$\mathbf{c}$}: $c_i=0$ for the $C$ (cache capacity) most popular contents according to $\mathbf{p}_0$, and 1 to the rest.

\noindent \emph{Solving \textbf{OP~\ref{problem:basis-LP},~\textbf{OP~\ref{problem:single}}}:} carried out using IBM ILOG CPLEX in Python. We note that since CPLEX is designed to receive LPs in the standard form, we had to vectorize our matrices in order to bring the problem in the format $\underset{\mathbf{x} \ge 0, \mathbf{A}\cdot \mathbf{x}\le \mathbf{b}}{\textbf{min}}~\{\mathbf{c}^{T}\cdot \mathbf{x}\}$ with linear and bound constraints over the variables. Regarding~\textbf{OP~\ref{problem:single}}, it is easy to see that the problem's objective Eq.(\ref{objective-single}) decomposes into $K$ independent minimization problems, of size $NK$ each, as the \emph{variables per content $i$ are not coupled}.
Finally note that for the simulations in Figs.~\ref{fig:no1},~\ref{fig:no2}, we will quote the cache-hit rate \emph{without} recommendations for reference, (i.e. storing the most popular contents that fit in the cache $C$, based on $\mathbf{p}_0$) and, 
which we denote as $MPH$ (Most Popular Hit - No Recommendations). This information along with the simulation parameters are included in Table~\ref{table:parameters}.

\subsection{Schemes we compare with}
We refer to our algorithm (\textbf{OP~\ref{problem:basis-LP}}) as \textit{Optimal}. 

\noindent \emph{Greedy Aware:} We consider as baseline algorithm for network-friendly recommendations
(\textbf{OP~\ref{problem:single}}~\cite{chatzieleftheriou2017technical}), which is a position-aware scheme, 
but does not take into account that requests are sequential. 

\noindent \emph{CARS:} algorithm~\cite{giannakas2018show}, a position-\emph{unaware} scheme for sequential content requests proposed in, will serve as our second baseline. The CARS algorithm optimizes (with no guarantees) the recommendations for a user performing multiple sequential requests, but assumes that the user selects \emph{uniformly} one of the recommendations regardless of the position they appear. 

\myitem{Note on \emph{CARS}.} In our framework, this translates to solving \textbf{OP~\ref{problem:basis}} for uniform $\mathbf{v}$. The algorithm will then return \emph{$N$ identical stochastic recommendation matrices}. Importantly, whichever $\mathbf{v}$ we choose, the parenthesis of the Eq.(\ref{eq:objective}) will be $(\mathbf{I}-\alpha \cdot (v_1\cdot \mathbf{R}+..+v_N\cdot \mathbf{R}))  = (\mathbf{I}-\alpha \cdot \mathbf{R})$. This explains why the hit rate of \emph{CARS} in the plots, remains constant regardless of the click distribution $\mathbf{v}$.

\subsection{Datasets}

\myitem{YouTube FR.} ($K = 1054$)
We used the crawling methodology of~\cite{kastanakis-cabaret-mecomm} and collected a dataset from YouTube in France. We considered 11 of the most popular videos on a given day, and did a breadth-first-search (up to depth 2) on the lists of related videos (max 50 per video) offered by the YouTube API. We built the matrix $\mathbf{U}\in\{0,1\}$ from the collected video relations.

\myitem{last.fm.} ($K = 757$)
We considered a dataset from the last.fm database \cite{lastfm}. We applied the ``getSimilar'' method to the content IDs' to fill the entries of the matrix $\mathbf{U}$ with similarity scores in [0,1]. We then set scores above $0.1$ to $u_{ij}=1$ to obtain a dense $\mathbf{U}$ matrix. 

\myitem{MovieLens.} ($K = 1066$) We consider the Movielens movies-rating dataset~\cite{movielens}, containing $69162$ ratings (0 to 5 stars) of $671$ users for $9066$ movies. We apply an item-to-item collaborative filtering (using 10 most similar items) to extract the missing user ratings, and then use the cosine distance ($\in[-1,1]$) of each pair of contents based on their common ratings. We set $u_{ij}=1$ for contents with cosine distance larger than $0.6$.

\subsection{Results}

\myitem{\emph{Optimal} vs \emph{CARS}.}
We initially focus on answering a basic question: \emph{Is the non-uniformity of users' preferences to some positions helpful or harmful for a network friendly recommender?}
In Figs.~\ref{fig:opt-agno-mvlns},~\ref{fig:relative-vs-cars} (see Table~\ref{table:parameters} for sim. parameters), we assume behaviors of increasing entropy; starting from users that show preference on the higher positions of the list (low entropy), to users that select uniformly recommendations (maximum entropy). In our simulations, we have used a zipf distribution~\cite{RecImpact-IMC10} over the $N$ positions and by decreasing its exponent, the entropy on the $x$-axis is increased. As an example, in Fig.~\ref{fig:opt-agno-mvlns}, lowest $H_\mathbf{v}$ corresponds to a vector of probabilities $\mathbf{v}=[0.8,0.2]$ (recall that $N=2$), while the highest one on the same plot to $\mathbf{v}=[0.58,0.42]$.

\myitem{Observation 1.} Our first observation is that the lower the entropy, the higher \emph{the optimal result}. 
In the extreme case where the $H_{\mathbf{v}} \to 0$ (virtually this would mean $N=1$, the user clicks deterministically), 
the optimal hit rate becomes maximum. 
This can be validated in
Fig.~\ref{fig:sensitivity-nb-of-recs}, 
where for increasing entropy the 
the hit rate decreases and its max is attained for $N=1$.





\begin{table}
\centering
\caption{Parameters of the simulation}\label{table:compare-youtube-people} \label{table:parameters}
\begin{small}
\begin{tabular}{l|l|l|l|l|l}
			& {q \%} & {$zipf(s)$} & {$\alpha$} & {$N$} & {MPH \%} \\
\hline \hline
{MovieLens}			&{80} &{0.8} &{0.7} & {2} & {23.26} \\
\hline
{YouTube FR}		&{95} &{0.6} &{0.8} & {2} & {12.17}\\
\hline
{last.fm}			&{80} &{0.6} &{0.7} & {3} & {11.74}\\

\end{tabular}
\end{small}
\end{table}

\begin{figure} 
\centering
\subfigure[Absolute Perf.]{\includegraphics[width=0.4\columnwidth]{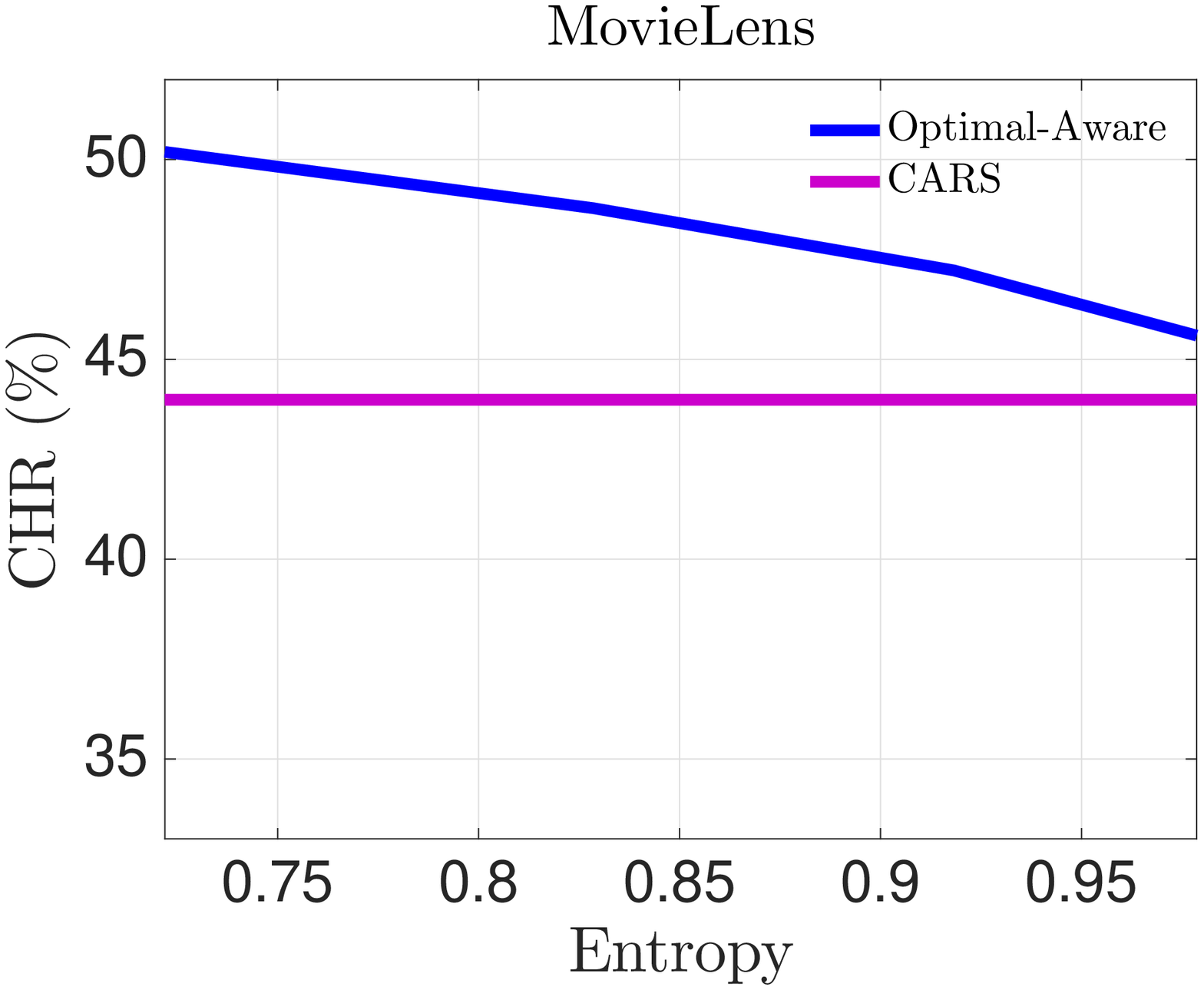}\label{fig:opt-agno-mvlns}} 
\hspace{0.08\columnwidth}
\subfigure[Relative Gain \%]{\includegraphics[width=0.4\columnwidth]{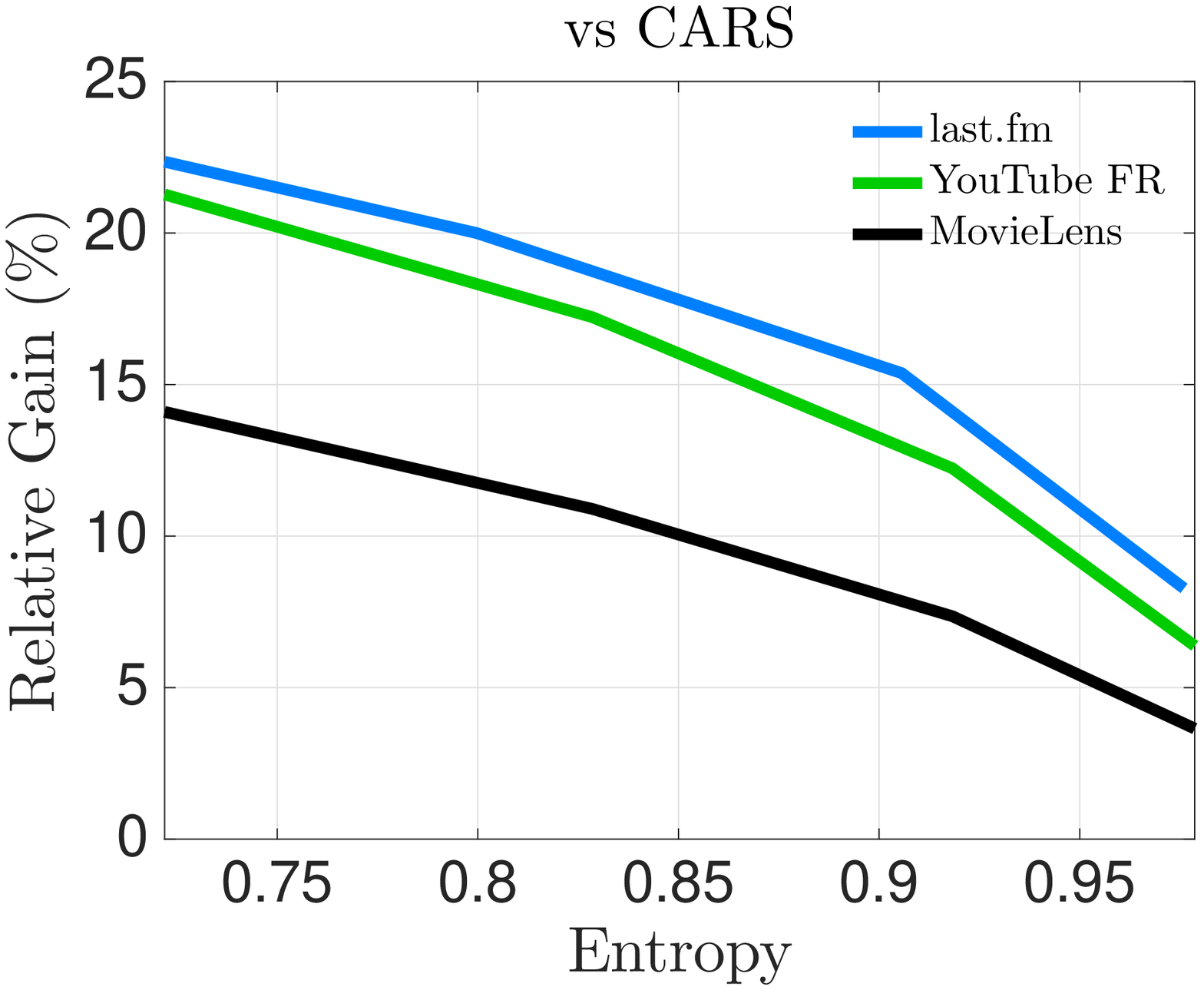}\label{fig:relative-vs-cars}}
\vspace{-2.5mm}
\caption{Cache Hit Rate vs~ $H_{\mathbf{v}}~(C/K\approx1.00\%)$}
\label{fig:no1}
\end{figure}

\begin{figure} \label{fig:opt-greedy}
\centering
\subfigure[Absolute Perf.]{\includegraphics[width=0.4\columnwidth]{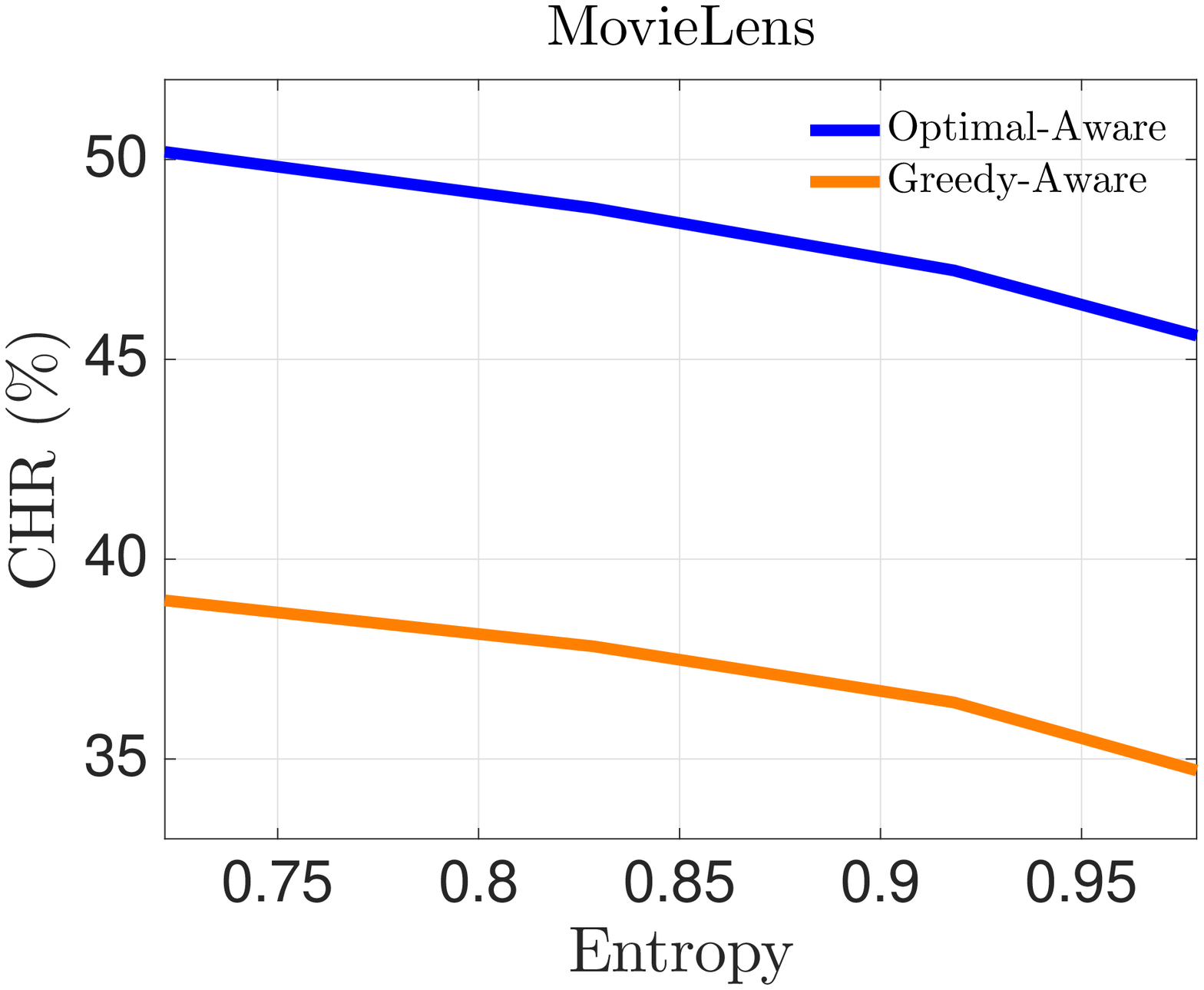}\label{fig:opt-greedy-mvlns}}
\hspace{0.08\columnwidth}
\subfigure[Relative Gain \%]{\includegraphics[width=0.4\columnwidth]{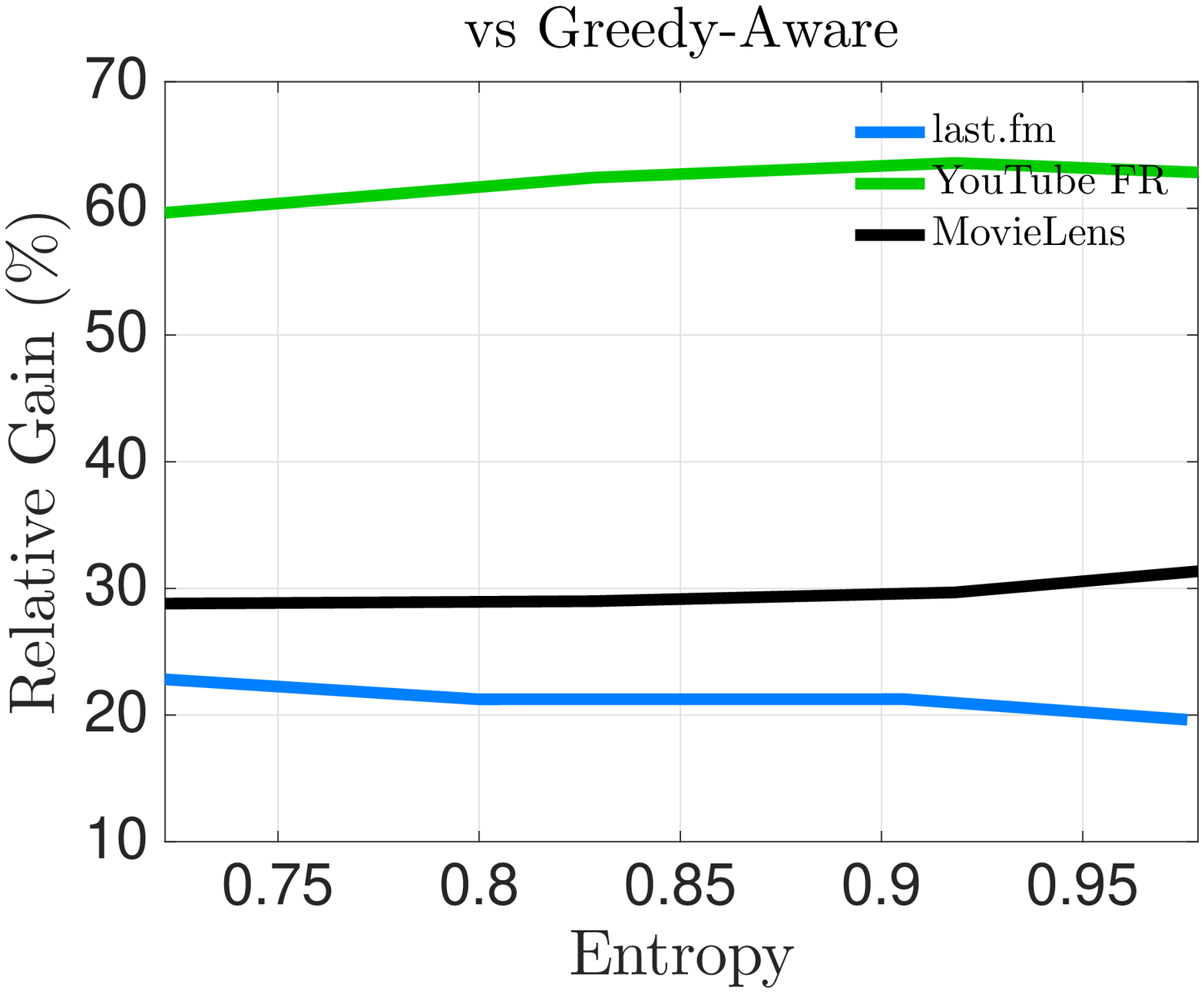}\label{fig:opt-greedy-relative-all}}
\vspace{-2.5mm}
\caption{Cache Hit Rate vs~ $H_{\mathbf{v}}~(C/K\approx1.00\%)$}
\label{fig:no2}
\end{figure}

\begin{figure} \label{fig:combination-plots}
\centering
\subfigure[$q=80\%,~K=400$]{\includegraphics[width=0.4\columnwidth]{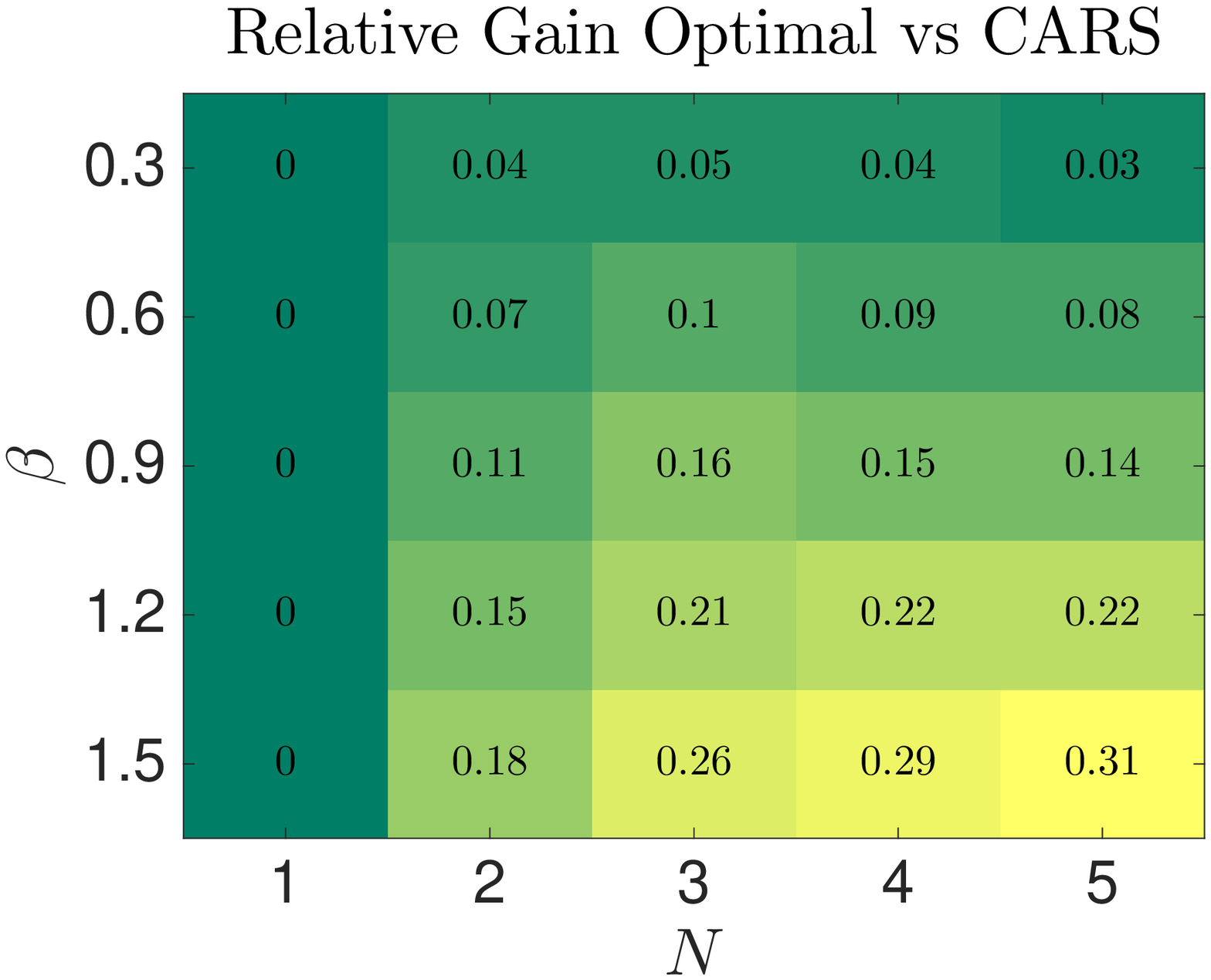}\label{fig:heatmap}}
\hspace{0.08\columnwidth}
\subfigure[Absolute Perf. ($q = 90\%,~s = 0.6,~MPH=11.24\%$)]{\includegraphics[width=0.4\columnwidth]{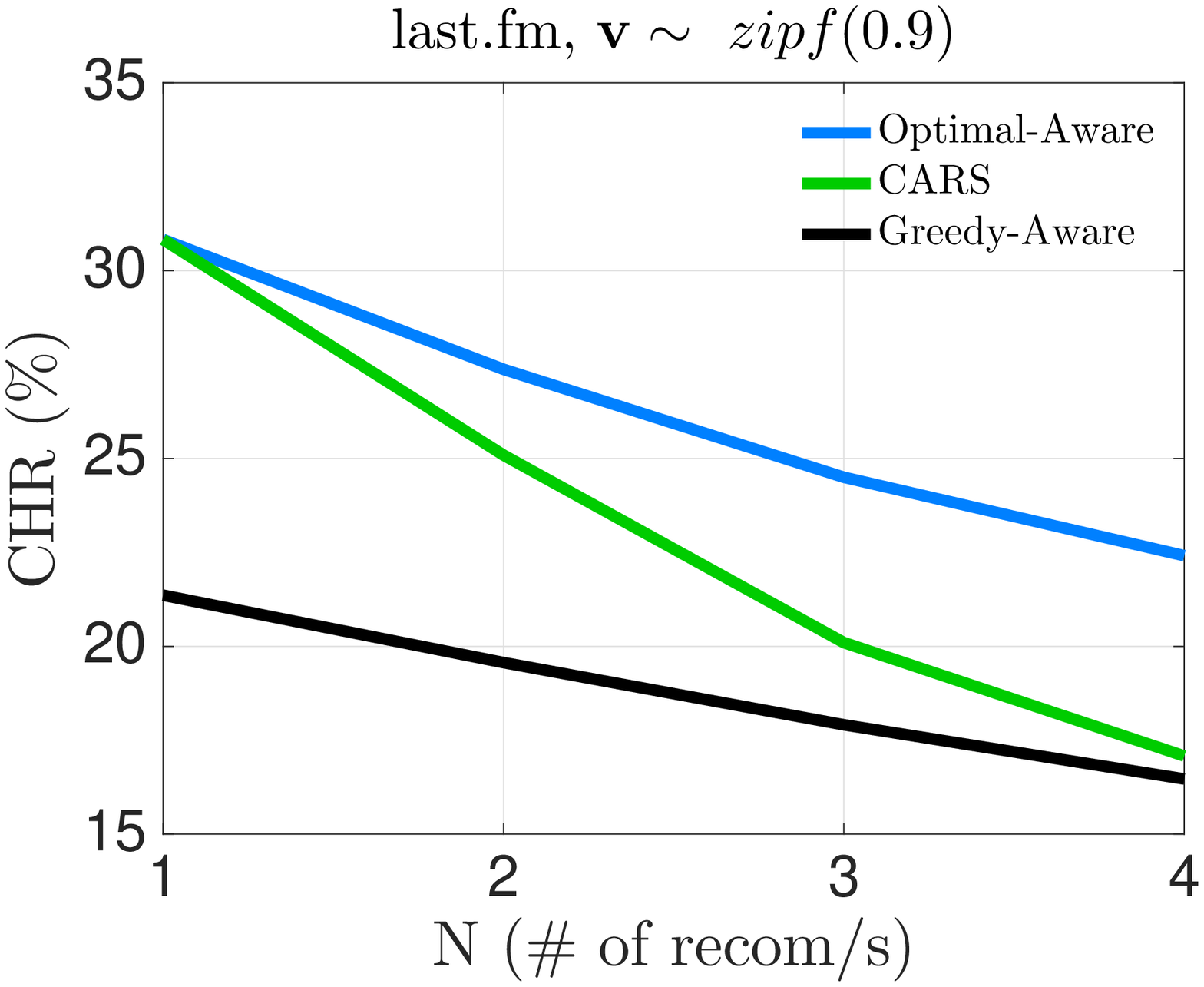}\label{fig:sensitivity-nb-of-recs}}
\vspace{-2.5mm}
\caption{(a:) Relative Gain vs $(N,\beta)$ and (b:) Cache Hit Rate vs~$N$ ~($C/K\approx 1.00\%,~\alpha=0.7$)}
\end{figure}

\myitem{\emph{Optimal} vs \emph{Greedy}.}
The second question we study is: \emph{How would a simpler greedy/myopic, yet position-aware, algorithm fare against our proposed method?} Fundamentally, the \emph{Greedy} algorithm solves a less constrained problem than \textbf{OP~\ref{problem:basis}}, and is therefore a more lightweight option in terms of execution time. However, the merits of using the proposed optimal method are noticeable in  Figs.~\ref{fig:opt-greedy-mvlns},~\ref{fig:opt-greedy-relative-all} (parameters in Table~\ref{table:parameters}). In all three datasets, we see an impressive improvement, between $20-60\%$. 

\myitem{Observation 2.} The constant relative gain of the two \emph{aware} algorithms hints that both, as the entropy increases, seem to do the right placement in the positions. However, as \emph{Greedy} decides with a small horizon, it cannot build the correct \emph{long} paths that lead to higher gains in the following requests (clicks) of the user.



Lastly, we investigate the sensitivity of the three methods against the number of recommendations ($N$)
.
In Fig.~\ref{fig:sensitivity-nb-of-recs}, we present the \emph{CHR} curves of all  three schemes for increasing $N$, where we keep constant the distribution $\mathbf{v} \sim zipf(0.9)$. As expected, for $N=1$ (e.g., YouTube autoplay scenario) \emph{CARS} and the proposed scheme coincide, as there is no flexibility in having \emph{only one} recommendation. However, as $N$ increases, \emph{CARS} and \emph{Greedy} decay at a much faster pace than the proposed scheme, which is more resilient to the increase of $N$. This leads to the following observation.

\myitem{Observation 3.} For large $N$, \emph{CARS} may offer the \say{correct} recommendations (cached or related or both), but it cannot place them in the right positions, as there are now too many available spots. 
In contrast, our algorithm \textit{Optimal} recommends the \say{correct} contents, and places the recommendations in the \say{correct} positions. 
Fig.~\ref{fig:heatmap}, strengthens even more the Observation 3; its key conclusion is that with high enough enough $\beta$ (i.e. low $H_{\mathbf{v}}$) and more than 2 or 3 recommendations, while \emph{CARS} aims to solve the multiple access problem, its \emph{position preference unawareness} leads to suboptimal recommendation placement, and thus severe drop of its \emph{CHR} performance compared to the \emph{Optimal}.



\section{Related Work}
\label{sec:related}
\myitem{RS and Caching Interplay.} The relation between RS and caching has only recently been considered~\cite{cache-centric-video-recommendation,chatzieleftheriou2017technical,content-recommendation-swarming,sermpezis2018soft,kastanakis-cabaret-mecomm,song2018making,lin2018joint,sch-chants-2016,liu2018learning}. 
Closer to our study, 
\cite{chatzieleftheriou2017technical} considers the joint problem of caching and recommendations, placing the most popular contents (among all users) in a cache and then trying to bias recommendations to favor cached contents, taking into account position preference in their model. However, this is applied to a different setup than ours (no markovian traversal of content graph); furthermore, they do not provide any simulation results on the impact of position preference.
The work in~\cite{giannakas2018show} tackles recommendations for users consuming multiple contents in a row, as we do. However,~\cite{giannakas2018show} formulates a nonconvex problem, and proposes a heuristic algorithm, 
and does not have optimality guarantees. 

\myitem{Optimization Methodology.} The problem of optimal recommendations for multi-content sessions, bares some similarity with PageRank manipulation~\cite{fercoq2013ergodic,avrachenkov2006effect,ermon2014designing}. The idea there is to choose the links of a subset of pages (the user has access to) with the intention to increase the PageRank of some targeted web page(s). Although that problem is generally hard, some versions of the problem can also be convexified~\cite{fercoq2013ergodic}. 


%
\section{Conclusions}
\label{sec:conclusions}
This work has proposed the optimal solution for network-friendly position aware recommendations. This technique can be used offline in a data-center of the content provider for network cost minimization.

Nevertheless, the area is still in its infancy. Theoretical and experimental research is needed to refine user behavior models and metrics, as well as dynamic learning and optimization of system parameters (e.g., user's reactivity to modified recommendations) or even where the ``network-friendly'' content appears in the recommendation list. Finally, jointly optimizing recommendations together with caching decisions (as in previous works~\cite{chatzieleftheriou2017technical}, but now  for multi-content sessions) is a key future step. 




\end{document}